\documentclass[a4paper]{article}
\usepackage{latexsym,amsthm,amsmath,amssymb}

\theoremstyle{plain}
\newtheorem{theorem}{Theorem}
\newtheorem*{theorem*}{Theorem}
\newtheorem{lemma}[theorem]{Lemma}
\newtheorem*{lemma*}{Lemma}
\newtheorem{corollary}[theorem]{Corollary}
\newtheorem{observation}[theorem]{Observation}
\newtheorem*{observation*}{Observation}

\newtheorem*{conjecture*}{Conjecture}

\theoremstyle{definition}
\newtheorem{definition}[theorem]{Definition}

\usepackage[utf8]{inputenc}
\usepackage[english]{babel}
\usepackage[numbers,sort]{natbib}
\usepackage{authblk}
\usepackage{graphicx}
\usepackage{subcaption}
\usepackage{pbox}
\usepackage{type1ec}
\usepackage[T1]{fontenc}
\usepackage{bbm}
\usepackage{mathtools}
\usepackage{amssymb}
\usepackage{amsmath}
\usepackage{amsthm}
\usepackage{amsfonts}
\usepackage{pifont}
\usepackage{url}
\usepackage{multirow}
\usepackage{multicol}
\usepackage{fancybox}
\usepackage{colortbl}
\usepackage{soul}
\usepackage{color}
\usepackage{tikz}
\usepackage{tkz-berge}
\usepackage{rotating}
\usepackage{pdflscape}
\usepackage{afterpage}
\usepackage{comment}
\usepackage{etoolbox}
\usepackage{environ}
\usepackage{dirtytalk}

\usepackage{algorithm}
\usepackage{algorithmic}

\usetikzlibrary{decorations,arrows,shapes}
\newcommand{\inners}{1.2pt}
\newcommand{\outers}{1pt}
\newcommand{\angled}[1]{\left\langle{#1}\right\rangle}

\usepackage{complexity}
\newclass{\Hard}{hard}
\newclass{\Hness}{hardness}
\newcommand{\NPH}{\NP-\Hard}
\newclass{\Complete}{complete}
\newclass{\Cness}{completeness}
\newcommand{\NPc}{\NP-\Complete}

\newfunc{\dist}{dist}
\newfunc{\diam}{diam}
\newfunc{\border}{border}
\newfunc{\supp}{supp}
\newfunc{\espan}{span}
\newfunc{\YES}{YES}
\newfunc{\NOi}{NO}
\newcommand{\pname}[1]{\textsc{#1}}

\newcommand{\bigO}[1]{\mathcal{O}\!\left(#1\right)}

\title{Some results on Vertex Separator Reconfiguration}
\author{Guilherme C. M. Gomes$^{2}$, S\'{e}rgio H. Nogueira$^{1,3}$, Vinicius F. dos Santos$^{1,2}$}
\affil{\small \textsuperscript{1}Centro Federal de Educação Tecnológica de Minas Gerais - Belo Horizonte, Brazil\\
\textsuperscript{2}Departamento de Ci\^{e}ncia da Computa\c{c}\~{a}o, Universidade Federal de Minas Gerais - Belo Horizonte, Brazil\\
\textsuperscript{3}Instituto de Ci\^{e}ncias Exatas e Tecnol\'{o}gicas, Universidade Federal de Vi\c{c}osa - Florestal, Brazil}
\date{}

\begin{document}

\maketitle

\begin{abstract} 
    We present the first results on the complexity of the reconfiguration of vertex separators under the three most popular rules: token addition/removal, token jumping, and token sliding.
    We show that, aside from some trivially negative instances, the first two rules are equivalent to each other and that, even if only on a subclass of bipartite graphs, TJ is not equivalent to the other two unless $\NP = \PSPACE$; we do this by showing a 
relationship between separators and independent sets in this subclass of bipartite graphs.
    In terms of polynomial time algorithms, we show that every class with a polynomially bounded number of minimal vertex separators admits an efficient algorithm under token jumping, then turn our attention to two classes that do not meet this condition: $\{3P_1, diamond\}$-free and series-parallel graphs.
    For the first, we describe a novel characterization, which we use to show that reconfiguring vertex separators under token jumping is always possible and that, under token sliding, it can be done in polynomial time; for series-parallel graphs, we also prove that reconfiguration is always possible under TJ and exhibit a polynomial time algorithm to construct the reconfiguration sequence.
\end{abstract}

% DO NOT REMOVE: Creates space for Elsevier logo, ScienceDirect logo
% and ENDM logo

\section{Introduction}
Reconfiguration problems have recently emerged in different fields of computer science, such as satisfiability~\cite{Gopalan, cnf_reconf},
constraint satisfaction~\cite{constraint_sat1, constraint_sat2}, computational geometry~\cite{geometry1, geometry2}, and quantum complexity theory~\cite{quantum}, even though reconfiguration questions have been posed in mathematics for well over a century~\cite{15_puzzle}.
More practically, in real world problems, it is often the case that the systems we work with have a preferred state or one of optimal performance, but currently find themselves out of such a configuration.
It is natural to try to answer questions such as: \textit{can we bring our system to the desired state without breaking it?} or \textit{what is the minimum number of tweaks we have to perform to do so?}

Not surprisingly, reconfiguration has found its way into graph theory, with work on the reconfiguration of classical graph structures recently appearing in the literature.
While \pname{Independent Set Reconfiguration} is by far the most well studied graph theoretical reconfiguration problem, with a corpus of at least a dozen papers~\cite{Kaminski,sliding_tokens,isr_cograph1,isr_cograph2,isr_claw_free,sliding_pspace_completeness,param_tj,param_reconf,td_reconf,bandwidth_reconf,Lokshtanov,Ito1}, 
many others have also attracted the attention of the community, such as \pname{Clique Reconfiguration}~\cite{Ito1, Ito2} and \pname{Vertex Coloring Reconfiguration}~\cite{Cereceda,chordal_color_reconf,coloring_reconf_pspace,td_reconf}.
Many of these works formalize their problems through the \textit{reconfiguration framework}.
Under this framework, a set of tokens is placed on some vertices of the input graph and is said to be a \textit{state} if it satisfies the desired property, e.g. induces an independent set.
Two states are \textit{adjacent} if we can apply an atomic operation to one in order to obtain the other, e.g. add a token to a vertex.
Three operations have been given considerable attention in the literature:

\begin{itemize}
    \item \textbf{Token Sliding (TS)}: Move the token on a vertex $v$ to any neighbor of $v$. 
    \item \textbf{Token Jumping (TJ)}: Move one of the tokens to any other vertex of the graph.
    \item \textbf{Token Addition/Removal (TAR)}: Add a token to any vertex of the graph, or remove a token from any vertex.
\end{itemize}

Under both TJ and TS, the cardinality of the token set remains unchanged throughout the whole reconfiguration process.
This is clearly not the case with TAR and the rule as is allows for trivial solutions to reconfiguration in many cases.
For instance, on \pname{Independent Set Reconfiguration} under TAR, we could remove one vertex of the initial set at a time until we have an empty independent set, then add each vertex of the target set, completing the reconfiguration.
To keep the question interesting and avoid this triviality phenomenon, a lower or upper bound is imposed on the cardinality of the token set; in our example, we would add the restriction that the intermediate independent sets must have \textit{at least} $k$ vertices.
Since the direction of the bound is usually clear by the problem definition, the bounded version of TAR is usually referred to as $k$\textit{-TAR}.
Throughout our text, unless the bound of $k$-TAR is relevant to the discussion, in an abuse of terminology, we use simply TAR.

Given a graph $G$, a property $\pi$, an operation $P$, we define the \textit{reconfiguration graph} $R_G(\pi, P)$, which has as vertex set the states defined by $\pi$ and an edge between two states if they are adjacent under $P$.
Under this formalism, we say that $A$ can be reconfigured into $B$ if the vertices corresponding to $A$ and $B$ in $R_G(\pi, P)$ are in the same connected component.
When discussing the complexity of a reconfiguration problem, the natural certificate is the sequence of states used to reconfigure $A$ into $B$.
This sequence, however, is not necessarily of polynomial size.
Thus, to show that a reconfiguration problem is \NPc, we must either use another certificate that can be verified in polynomial time, or we must prove that \YES\ instances have a reconfiguration sequence of polynomial length.
For an example of the latter, we refer to the work of \citeauthor{Lokshtanov}~\cite{Lokshtanov} on \pname{Vertex Cover Reconfiguration}, more specifically their result on bipartite graphs under TAR.

In this work, we explore the reconfiguration problem of another fundamental structure of graph theory: \textit{vertex separators}; in this version of the problem, vertices $s,t$ we wish to separate are fixed, i.e. a set $S \subseteq V(G)$ satisfies the property of interest $\pi$ if and only if, in $G - S$, $s$ and $t$ are disconnected.
We refer to this problem as \pname{Vertex Separator Reconfiguration}.
We begin by showing that, aside from a set of easily identifiable negative instances.
In terms of hardness results, we show that, on a subclass of bipartite graphs, the problem is \PSPACE-\Complete\ under TS and \NPc\ under TAR/TJ.
On the positive side, we first show that, if a graph class  is tame, i.e. has a polynomially bounded number of minimal separators, under TAR/TJ, \pname{Vertex Separator Reconfiguration} can be solved in polynomial time.
Then, we turn our attention to classes that do not satisfy this condition.
In particular, we present a characterization for $\{3P_1,\text{diamond}\}$-free graphs, which we use to show that reconfiguration is always possible under TAR/TJ and easily verifiable under TS.
Our final results concern the class of $2$-connected series-parallel graphs, to which we provide a thorough analysis of the possible $st$-separators, which we use to show that \pname{Vertex Separator Reconfiguration} is also always positive under TAR/TJ, leaving the question under TS as future work.
Other questions we leave unanswered include the study of other non-tame classes, particularly the ones described by \citeauthor{Milanic}~\cite{Milanic}, and the search for reasonable sufficient conditions for the existence of polynomial time algorithms under TS.

\section{Preliminaries}

In this work we deal with finite simple graphs, and follow the terminology of standard textbooks in the area, such as~\cite{murty}. 
Given a connected graph $G$, we denote by $\diam(G)$ the length of the longest shortest path between any two vertices of $G$.
Given two non-adjacent vertices $s,t \in V(G)$, we say that $S \subseteq V(G) \setminus \{s,t\}$ is an $st$-separator if, on the subgraph $G - S$ induced by $V(G) \setminus S$, there is no path between $s$ and $t$; $S$ is minimal if no proper subset of $S$ is also an $st$-separator.
We say that a sequence $ \mathcal{S} = \langle S_1, \dots, S_n \rangle$ of subsets of $V(G)$ is a \textit{reconfiguration sequence} if, for all $i$, $S_i$ is an $st$-separator.
$\mathcal{S}$ is a \textit{reconfiguration sequence under $P$} if, for all $i \in [n-1]$, we can obtain $S_{i+1}$ by applying operation $P$ to $S_i$ once.
For example, if our operation of interest is token sliding, we must replace one of the vertices in $S_i$ by one of its neighbors.
Formally, for our three rules of interest, we say that two separators $S_i$, $S_j$ are \textit{adjacent}, denoted by $S_i \leftrightarrow S_j$, if the following holds:

\begin{itemize}
	\item[\textbf{TS}] $S_i \leftrightarrow S_j$ if  $| S_i | =| S_j |$, $ S_i  \setminus S_j =\{u_i\}$, $ S_j  \setminus S_i=\{v_j\}$ and $u_iv_j \in E(G)$.
	
	\item[\textbf{TJ}] $S_i \leftrightarrow S_j$ if  $| S_i | = | S_j |$ and $| S_i  \setminus S_j |= | S_j  \setminus S_i|=1$.
	
	\item[\textbf{$k$-TAR}]$S_i \leftrightarrow S_j$ if  $| S_i \Delta S_j |=|( S_i \setminus S_j) \cup ( S_j \setminus S_i)  |=1$ and $\max(|S_i|, |S_j|) \leq k$.
\end{itemize}

When describing our instances, we adopt the convention that $TS(G, S_a, S_b)$ means that we want to reconfigure $S_a$ into $S_b$ on graph $G$ under TS; for $k$-TAR, we use $TAR(G, S_a, S_b, k)$.
Following the notation previously used in the literature, if the instance is positive, i.e. $S_a$ can be reconfigured into $S_b$, we use the notation $S_a \leftrightsquigarrow S_b$.
\section{TAR-TJ equivalence}

During their study of \pname{Clique Reconfiguration}, \citeauthor{Ito2}~\cite{Ito2} proved that rules TAR, TJ
and TS are equivalent in terms of the complexity of the problem and that, in fact, the length of the reconfiguration sequences are within a factor of two from each other.
Inspired by their work, our first result states that \pname{Vertex Separator Reconfiguration} is equivalent under TAR and TJ, in the sense that, given a graph $G$ and two vertex separators  $S_a,S_b$ in $G$,  if we have a $TJ$-instance $(G,S_a,S_b)$  we can construct a $TAR-$instance such that $TJ(G,S_a,S_b)=TAR(S_a',S_b',k)$ and if we have  a $TAR$-instance we can construct a $TJ-$instance such that $TAR(G,S_a,S_b, K)=TJ(S_a',S_b')$.
Specifically, we show that every TJ instance has a corresponding TAR instance and, in the other direction, either the TAR instance is trivially negative or there is some TJ instance to which it is equivalent to.
In the following lemma, we show that, if $S_a, S_b$ are two minimal $uv$-separators such that $\vert S_a \vert = \vert S_b \vert$ and $S_a \leftrightsquigarrow S_b$ under $TAR$, there exists a shortest $TAR$ sequence that can be constructed in linear time adding and removing a vertex in each step. 

\begin{lemma}
    \label{lem:canon_tar}
	Let $(G, S_a, S_b, k+1)$ be a $TAR$-instance of \pname{Vertex Separator Reconfiguration}, where $S_a, S_b$ are a pair of $uv$-separators satisfying $|S_a| = |S_b| = k$.
	If $S_a \leftrightsquigarrow S_b$, there exists a shortest reconfiguration sequence $\langle S_a = S_1,..., S_\ell = S_b \rangle$ such that $|S_i| \geq k$ for all $i \in [\ell]$, with equality holding if and only if $i$ is odd.
\end{lemma}

\begin{proof}
    Let $\mathcal{S}$ be a shortest reconfiguration sequence for $TAR(G, S_a, S_b, k+1)$ of length $\ell$.
    We modify $\mathcal{S}$ by the following algorithm that, at each step, selects exactly one state of the sequence and increases its size by two, while $\min \{|S_i| \mid i \in [\ell]\} < k$.
    If the algorithm is not done, take $j$ as the smallest index such that $|S_j| = \min \{|S_i| \mid i \in [\ell]\}$.
    Since $j \notin \{1, \ell\}$, $S_j = S_{j-1} \setminus \{a\}$ and  $S_j = S_{j+1} \setminus \{b\}$, and $a \neq b$ because $\mathcal{S}$ is a shortest sequence.
    Furthermore, it holds that $|S_j \cup \{a,b\}| \leq k+1$ and we can construct another sequence $\mathcal{S}' = \angled{S_1, \dots, S_{j-1}, S_j \cup \{a,b\}, S_{j+1}, \dots S_\ell}$ and set $\mathcal{S} \gets \mathcal{S}'$.
    When the algorithm stops, every $S_i$ satisfies $|S_i| \in \{k, k+1\}$ and, since it never modifies $S_1$ nor $S_\ell$ and every operation modifies the size of the state it acts upon, $|S_2| = k+1$; a direct inductive argument shows that $|S_{2i}| = k+1$ for all $i \in [\ell/2]$.
\end{proof}

In the following, let $\dist_{X}(G, S_a, S_b)$ be the size of a shortes reconfiguration sequence from $S_a$ to $S_b$ according to rule $X$.

\begin{lemma} \label{lem:tj_tar}
	Let $G$ be a graph and let $S_a, S_b$ be any pair of $uv$-separators in $G$ such that $|S_a| = |S_b| = k$.
	It holds $TJ(G, S_a, S_b) = TAR(G, S_a, S_b, k+1)$ and $\dist_{TJ}(G, S_a, S_b) = \dist_{(k+1)-TAR}(G, S_a,S_b)/2$.
\end{lemma}

\begin{proof}
    If $\mathcal{S}_{TJ} = \angled{S_1 = S_a, \dots, S_\ell = S_b}$ is a shortest TJ-reconfiguration sequence, we construct the sequence $\mathcal{Q}_{TAR} = \angled{Q_1, \dots, Q_r}$ such that $Q_{2i-1} \gets S_i$ for all $i \in [\ell]$ and, for all even $j \in [2\ell - 1]$, $Q_j = Q_{j-1} \cup Q_{j+1} = S_{j/2} \cup S_{j/2 + 1}$.
    Note that, in the latter case, $|Q_j| \leq k+1$, since $S_{j/2}$ and $S_{j/2 + 1}$ are consecutive states of $\mathcal{S}_{TJ}$; furthermore, $Q_j \setminus Q_{j+1}$ is empty if $j$ is even or has cardinality equal to one if $j$ is odd.
    As such, we conclude $\mathcal{Q}_{TAR}$ is a certificate for $TAR(G, S_a, S_b, k+1)$ and that $\dist_{(k+1)-TAR}(G, S_a, S_b) \leq 2\dist_{TJ}(G, S_a, S_b)$.
    
    For the converse, let $\mathcal{S}_{TAR} = \angled{S_1 = S_a, \dots, S_\ell = S_b}$ be a shortest $(k+1)$-TAR-reconfiguration sequence satisfying Lemma~\ref{lem:canon_tar}.
    We show that the sequence $\mathcal{Q}_{TJ} = \angled{Q_1, \dots, Q_r}$ where $Q_{j} \gets S_{2j-1}$, for all $j$ satisfying $1 \leq 2j -1 \leq \ell$, is a TJ-reconfiguration sequence.
    By Lemma~\ref{lem:canon_tar}, it holds that $|Q_j| = k$ and $Q_{j} \triangle Q_{j+1} = S_{2j-1} \triangle S_{2j+1} = \{a,b\}$ and the jump operation is executed in such a way that, say, $a$ is replaced by $b$.
    Moreover, $\dist_{TJ}(G, S_a, S_b) \leq \dist_{(k+1)-TAR}(G, S_a, S_b)/2$, concluding the proof.
\end{proof}

Observation~\ref{obs:tar_no} states that the value of $k+1$ imposed by the hypothesis of Lemma~\ref{lem:tj_tar} is optimal, i.e. there are instances $TJ(G, S_a, S_b) = \YES$ where at least one of the endpoints of the reconfiguration is a minimal separator but $TAR(G, S_a, S_b, k) = \NOi$.
Furthermore, it prunes trivially negative cases which, as we show further below, are the only ones that may not have a corresponding TJ instance.

\begin{observation}
    \label{obs:tar_no}
	Let $TAR(G, S_a,S_b,k)$ be an instance of \pname{Vertex Separator reconfiguration} such that $S_a \neq S_b$. If at least one of $S_a, S_b$ is minimal and has exactly $k$ vertices, then $TAR(G, S_a,S_b,k)= \NOi$.
\end{observation}

\begin{lemma}
    \label{lem:tar_tj}
    Let $TAR(G, S_a, S_b, k)$ be an instance of \pname{Vertex Separator Reconfiguration} that does not satisfy Observation~\ref{obs:tar_no}, $S_a'$ be a set of cardinality $k-1$ such that $S_a' \cap S_a$ contains a $uv$-separator and $S_b'$ be defined analogously.
    It holds that $TAR(G, S_a, S_b, k) = TAR(G, S_a', S_b', k) = TJ(G, S_a', S_b')$ and $\dist_{TJ}(G, S_a', S_b') = \dist_{k-TAR}(G, S_a',S_b')/2$.
\end{lemma}

\begin{proof}
    To see that the first equality holds, note that $S_a$ ($S_b$) can be easily reconfigured into $S_a'$ ($S_b'$), since both $S_a$ ($S_b$) and $S_a'$ ($S_b'$) contain the \textit{same} $uv$-separator of $G$.
    For the second, note that the triple $(S_a', S_b', k)$ satisfy the hypotheses of Lemma~\ref{lem:tj_tar} precisely because neither $S_a$ nor $S_b$ are minimal and of size $k$.
\end{proof}

The direct application of Lemmas~\ref{lem:tj_tar}, \ref{lem:tar_tj}, and Observation~\ref{obs:tar_no} yield our equivalence theorem.

\begin{theorem}
    \label{thm:tj_tar_eq}
	Let $G$ be a graph,  $u,v \in V(G)$, and $S_a, S_b$ two $uv$-separators. The following statements hold:
	\begin{itemize}
		\item [(i)] If $|S_a| = |S_b| = k$, instance $TJ(G, S_a, S_b)$ has an equivalent TAR-instance $TAR(G, S_a, S_b, k+1)$ which can be built in linear time and $\dist_{TJ}(G, S_a, S_b) = \dist_{(k+1)-TAR}(G, S_a,S_b)/2$;
		\item[(ii)] If $|S_a| \leq |S_b| = k$ and $S_b$ is minimal, $TAR(G, S_a, S_b, k) = \NOi$;
		\item[(iii)] Otherwise, instance $TAR(G, S_a, S_b, k)$ has an equivalent TJ-instance $TJ(G, S_a', S_b')$ that can be generated in polynomial time satisfying $\dist_{TJ}(G, S_a', S_b') = \dist_{k-TAR}(G, S_a',S_b')/2$.
	\end{itemize}
\end{theorem}

We cannot ascertain whether TS is equivalent or not to TAR/TJ in a similar sense as to the one given in Theorem~\ref{thm:tj_tar_eq}.
It is not surprising, however, that we cannot do so: as we show in Corollary~\ref{cor:completeness} of Theorem~\ref{thm:hardnesses}, if all three rules are equivalent even if only on a very specific subclass of bipartite graphs, we would have that $\NP = \PSPACE$.
Figure~\ref{fig:ts_tar_diff} presents an example where it is not immediate that one cannot reconfigure $S_a = \{u_1, u_2\}$ into $S_b = \{u_7, u_8\}$ under TS, but we can do so under TAR/TJ.

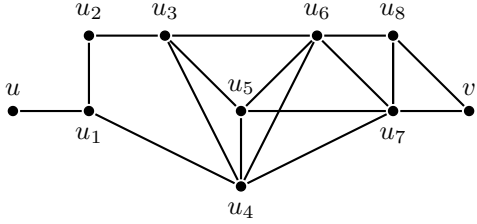
\begin{figure}[!htb]
    \centering
    \begin{tikzpicture}[scale=1]
            %\draw[help lines] (-5,-5) grid (5,5);
            \GraphInit[unit=3,vstyle=Normal]
            \SetVertexNormal[Shape=circle, FillColor=black, MinSize=2pt]
            \tikzset{VertexStyle/.append style = {inner sep = \inners, outer sep = \outers}}
            \SetVertexLabelOut
            \Vertex[x=0, y = 0, Lpos=90, Math]{u_5}
            \Vertex[x=0, y = -1, Lpos=270, Math]{u_4}
            \Vertex[x=-1, y = 1, Lpos=90, Math]{u_3}
            \Vertex[x=1, y = 1, Lpos=90, Math]{u_6}
            
            \Vertex[x=-3, y = 0, Lpos=90, Math]{u}
            \Vertex[x=-2, y = 0, Lpos=270, Math]{u_1}
            \Vertex[x=-2, y = 1, Lpos=90, Math]{u_2}
            
            \Vertex[x=2, y = 0, Lpos=270, Math]{u_7}
            \Vertex[x=2, y = 1, Lpos=90, Math]{u_8}
            \Vertex[x=3, y = 0, Lpos=90, Math]{v}
            \Edges(u_4, u_1, u_2, u_3, u_5, u_6, u_3, u_4, u_6, u_7, u_5, u_4)
            \Edges(u, u_1)
            \Edges(u_4, u_7, u_8)
            \Edges(u_6, u_8, v, u_7, u_8)
    \end{tikzpicture}
    \caption{Under TAR/TJ, we can reconfigure $S_a = \{u_1, u_2\}$ into $S_b = \{u_7, u_8\}$, but not under TS.\label{fig:ts_tar_diff}}
\end{figure}
\section{Hardness results}

\citeauthor{Lokshtanov}~\cite{Lokshtanov} proved that \pname{Independent Set Reconfiguration} is \PSPACE-\Complete\ under TS and \NPc\ under TAR/TJ for bipartite graphs.
In this section, we present an  equivalence between independent sets and vertex separators in a subclass of bipartite graphs and, using this relationship, give a reduction from \pname{Independent Set Reconfiguration} on bipartite graphs to \pname{Vertex Separator Reconfiguration} on the same class, concluding that the latter is \PSPACE-\Hard\ under TS and \NP-\Hard \ under TAR/TJ for bipartite graphs.
We say that a bipartite graph is \textit{peanut-like} if there is a pair of vertices $u \in B$, $v \in A$ such that $N(u) = A \setminus \{v\}$ and $N(v) = B \setminus \{u\}$; in this case we say that $u$ and $v$ are the \textit{foci} of $G$.
Our reductions show that, for peanut-like bipartite graphs, \pname{Vertex Separator Reconfiguration} is \NPc\ under TAR/TJ and \PSPACE-\Complete\ under TS.

\begin{lemma}\label{lem:stable_trans}
	Let $G=(A \cup B, E)$ be a bipartite graph with partition $A,B$, $H$ be the graph constructed from $G$ by adding two vertices $u$ and $v$ to $G$ such that $N(u)= A$ and $N(v)=B$. A set $I \subset V(G)$ is independent if and only if $V(G) \setminus I$ is a $uv$-separator in $H$.
\end{lemma}

\begin{proof}
	Let $I$ be  an independent set of $G$.
	Towards a contradiction, suppose that $V(G)\setminus I$ is not a $uv$-separator in $H$: thus there is some path $\{u, x, y, v\}$ between $u$ and $v$ in $H$ such that $x,y \notin V(G) \setminus I$, implying that $I$ is not independent.
	
	Conversely, let $I$ be a subset of $V(G)$ such that $V(G) \setminus I$ is a $uv$-separator in $H$ but suppose that $I$ is not an independent set of $G$, i.e., there are two adjacent vertices $x,y \in I$ with $x \in A$ and $y \in B$.
	Then there is a path $\{u, x, y, v\}$ from $u$ to $v$ in $H - (V(G) \setminus I)$, contradicting the hypothesis that $V(G) \setminus I$ is a $uv$-separator of $H$. 	
\end{proof}

\begin{corollary}
    Let $ H $ be an $n$-vertex peanut-like bipartite graph with $u,v$ as its foci and $S_a, S_b$ two $uv$-separators of $H$.
    If $S_a$ can be reconfigured into $S_b$, then there is a reconfiguration sequence between them of length $\bigO{n^4}$.
\end{corollary}

\begin{proof}
    Let $G = H - \{u,v\}$; by Lemma~\ref{lem:stable_trans}, there is a one to one correspondence between $uv$-separators of $H$ and independent sets of $G$.
    Furthermore, by Theorem 3 of~\cite{Lokshtanov}, there is a reconfiguration sequence between $V(G) \setminus S_a$ and $V(G) \setminus S_b$ if and only if there is some sequence of length $\bigO{|V(G)|^4}$; again by Lemma~\ref{lem:stable_trans}, this implies that there is a reconfiguration sequence between $S_a$ and $S_b$ of polynomial length.
\end{proof}

\begin{theorem}
    \label{thm:hardnesses}
	\pname{Vertex Separator Reconfiguration} on bipartite graphs, under TAR/TJ is \NP-\Hard\  and \PSPACE-\Hard\ under TS.
\end{theorem}

\begin{proof}
	Our reduction is from \pname{Independent Set Reconfiguration} under TJ on bipartite graphs, shown to be \NP-\Hard\ in~\cite{Lokshtanov}.
	Let $(G, I_a, I_b)$ an $ISR_{TJ}$-instance of \pname{Independent Set Reconfiguration} under $TJ$; our  $(H, S_a, S_b)$ $TJ$-instance of \pname{Vertex Separator Reconfiguration} is built as follows:
	graph $H$ is defined by $V(H) = V(G) \cup \{u,v\}$ and $E(H) = \{ua \mid a \in A\} \cup \{bv \mid b \in B\} \cup E(G)$, $S_a = V(G) \setminus I_a$, and $S_b = V(G) \setminus I_b$.
	
	Now suppose that $\langle I_1, \dots, I_r \rangle$ is a reconfiguration sequence of independent sets of $G$, we construct the sequence $\langle S_1, \dots, S_r \rangle$ by setting $S_i = V(G) \setminus I_i$.
	By Lemma~\ref{lem:stable_trans}, each $S_i$ is a $uv$-separator, $S_1 = S_a$, $S_r = S_b$ $|S_i| = |S_{i+1}|$, moreover, since there is exactly one $z \in I_i \setminus I_{i+1}$, exactly one $w \in I_{i+1} \setminus I_i$ and $S_i = V(G) \setminus I_i$, it holds that $z \in S_{i+1} \setminus S_i$ and $w \in S_i \setminus S_{i+1}$, showing that $\langle S_1, \dots, S_r, \rangle$ is a $uv$-separator reconfiguration sequence.
	The converse follows the exact same argumentation, and we omit it for brevity.
	
	Due to the TAR/TJ equivalence given by Theorem~\ref{thm:tj_tar_eq}, it also holds that \pname{Vertex Separator Reconfiguration} is  \NP-\Hard\ under TAR on bipartite graphs.
	For TS, given an instance $ISR_{TS}(G, I_a, I_b)$ of \pname{Independent Set Reconfiguration} under TS, which is \PSPACE-\Hard, we proceed exactly as we did when considering TJ and note that the unique vertices $z \in I_i \setminus I_{i+1}$ and $w \in I_{i+1} \setminus I_i$ are adjacent in $H$ since $H[V(G)]$ is isomorphic to $G$.
\end{proof}

\begin{corollary}
    \label{cor:completeness}
    \pname{Vertex Separator Reconfiguration} on peanut-like bipartite graphs under TAR/TJ is \NPc\ and under TS is \PSPACE-\Complete.
\end{corollary}

\begin{proof}
    For the first statement, the instance constructed on the proof of Theorem~\ref{thm:hardnesses} is a peanut-like bipartite graph, and an algorithm that inspects each state to decide whether a given sequence is a TAR/TJ reconfiguration sequence or not suffices to prove membership in \NP.
    For TS, on the other hand, we execute the following algorithm: given a peanut-like bipartite graph, we remove its foci and use the polynomial space algorithm for \pname{Independent Set Reconfiguration}, which outputs a reconfiguration sequence if and only if the answer to the original \pname{Vertex Separator Reconfiguration} instance is \YES.
\end{proof}

\section{Polynomial time results}

\citeauthor{Milanic}~\cite{Milanic} studied the behavior of the family of minimal vertex separators on graph classes defined by forbidden families of small induced subgraphs.
Using their nomenclature, a graph class $\mathcal{G}$ is \textit{tame} if the family of minimal vertex separators of each $G \in \mathcal{G}$, denoted by $\mathbf{S}$, has size bounded by a polynomial $p_{\mathcal{G}}$ evaluated at $|V(G)|$.
The opening result of this section states that if $G$ belongs to a tame class, then \pname{Vertex Separator Reconfiguration} is solvable in polynomial time.

\begin{lemma} \label{lem:overlap_graph}
    Let $u,v$ be two vertices of $G$, $\mathbf{S}_{uv}(G)$ be the family of minimal $uv$-separators of $G$, and
    $H_{uv}$ be the graph where $V(H_{uv})= \mathbf{S}_{uv}(G)$ and $E(H_{uv})=\{S_i,S_j \ \ \vert \  \  S_i,S_j \in  \mathbf{S}_{uv}(G) \ \ and \ \ \vert S_i\cup S_j \vert \leq k \}$.
    For any two $uv$-separators $S_a, S_b$ of $G$, $TAR(G, S_a,S_b,k) = \YES$ if and only if there exists a path from $S_a'$ to $S_b'$ in $H_{uv}$, where $S_a'$ and  $S_b'$ are minimal $uv$-separators of $G$ with $S_a' \subseteq S_a$ and $S_b' \subseteq S_b$.
\end{lemma}

\begin{proof}
    Let  $S_a, S_b$ two $uv$-separators in $G$.
    Suppose that $TAR(G, S_a, S_b, k) = \YES$, let $\langle S_1, \dots, S_r \rangle$ be a reconfiguration sequence between $S_a$ and $S_b$, and let $\langle \mathcal{S}_1, \dots \mathcal{S}_r\rangle$ be a sequence such that $\mathcal{S}_i$ is the family of all minimal $uv$-separators that are subsets of $S_i$.
    Note that $\mathcal{S}_i \cup \mathcal{S}_{i+1}$ is a clique of $H_{uv}$ since, for any $A \in \mathcal{S}_i$ and $B \in \mathcal{S}_{i+1}$, $|A \cup B| \leq |S_i \cup S_{i+1}| \leq k$.
    Thus, there is a path between $S_a'$ and $S_b'$ in $H_{uv}$.
    %{\color {red} muito boa, bem mais elegante} {\color {blue} valeu!}  
    
    For the converse, let $\langle S_a', S_1', \dots, S_r', S_b' \rangle$ be some path between $S_a'$ and $S_b'$ in $H_{uv}$; note that since $|S_i' \cup S_{i+1}'| \leq k$, we can greedily reconfigure $S_i'$ into $S_{i+1}'$ without violating the cardinality constraint; by a straightforward inductive argument, we can reconfigure $S_a'$ into $S_b'$ and, consequently, $S_a \leftrightsquigarrow S_b$.
\end{proof}

\begin{theorem} \label{separadorespolinomial}
	If $G$ has a polynomially bounded number of minimal vertex separators, then \pname{Vertex Separator Reconfiguration} can be solved in polynomial time under TAR/TJ. 
\end{theorem}

\begin{proof}
    Along with the result of \citeauthor{separator_generation}~\cite{separator_generation} that the family of minimal separators of an $n$-vertex graph $G$ can be generated in $\bigO{|\mathbf{S}_{uv}|n^3}$, Lemma~\ref{lem:overlap_graph} directly implies that it suffices to construct $H_{uv}$ and check if there is some minimal separator contained in $S_a$ in the same connected component of a minimal separator contained in $S_b$.
    Since $H_{uv}$ has a number of vertices polynomial on the size of $G$, this algorithm runs in time polynomial in $n$.
\end{proof}

\begin{corollary}
	\pname{Vertex Separator Reconfiguration} can be solved in polynomial time for chordal graphs under TAR/TJ.
\end{corollary}

\subsection{Non-tame classes}

Also in~\cite{Milanic}, \citeauthor{Milanic} determined the three families of graphs on at most four vertices that, when forbidden, do not yield a tame class.
Specifically, if $\mathcal{F} \in \{\{3P_1, diamond\}, \{claw, K_4, C_4, diamond\}, \{K_3, C_4\}\}$, then the class of $\mathcal{F}$-free graphs is \textit{not} tame; all other graph classes that exclude graphs on at most four vertices are tame.
In this final section of the paper, we show that not only can we solve \pname{Vertex Separator Reconfiguration} in polynomial time on  $\{3P_1, diamond\}$-free graphs and series-parallel graphs, but that reconfiguration is \textit{always} possible under TAR/TJ for these classes.
For the first, we show that, under TS, we can decide in polynomial time whether reconfiguration is possible.

\subsubsection{$\boldsymbol{\{3P_1, diamond\}}$-free}

Before dealing with the reconfiguration problem for $\{3P_1, diamond\}$-free graphs, we present a novel characterization of the class that makes the reconfiguration question almost trivial.
 
\begin{theorem}
    \label{thm:char}
	Let $G$ be a connected not complete graph with at least four vertices %{\color{red} não é bem isso, certo? }
	such that $G$ is not a clique and $G \neq C_5$. 
	Then $G$ is $\{3P_1, diamond\}$-free if and only if $\diam(G) \leq 3$ and one of the following statements hold:

	\begin{itemize}
		\item[(i)] $G$ is the union of two cliques $Q_1, Q_2$ and has exactly one cut vertex; or
		\item[(ii)] $G$ is the disjoint union of two cliques $Q_1, Q_2$ and the edges between the cliques form a matching.
	\end{itemize}
\end{theorem}

\begin{proof}
    For the forward direction, suppose $G$ has a universal vertex $v$; since $G$ is not a clique, $\diam(G) = 2$, and let $\{u,v,w\}$ be an induced $P_3$ of $G$.
    Since $G$ is $3P_1$-free, any $z \in V(G)$ is either in $N(u)$ or $N(v)$ but, since $G$ is also diamond-free and $v$ is universal, $z$ is not in both, otherwise $\{u,v,w,z\}$ is a diamond; moreover, if there are two non-adjacent vertices $a,b \in N(u)$, $\{a,u,v,b\}$ is an induced diamond, so $N[u]$ and $N[w]$ are cliques.1 Thus, $G$ satisfies condition (i) with $v$ as a cut vertex.
    
    If $G$ does not have a universal vertex, let $u,w \in V(G)$ be a pair of non-adjacent vertices and note that $N(u) \cup N(w) \cup \{u,w\} = V(G)$, otherwise $G$ would not be $3P_1$-free.
    Furthermore, there are no two adjacent vertices in $N(u) \cap N(w)$, otherwise we would have a diamond, $N(u) \setminus N(w)$ is a clique, otherwise $G\left[\{w\} \cup N(u) \setminus N(w)\right]$ would contain $3P_1$ as an induced subgraph, and no vertex in  $N(u) \setminus N(w)$ is adjacent to more than one vertex in $N(w)$, otherwise we would have an induced diamond.
    We branch our analysis on the size of $N(u) \cap N(w)$.
    \begin{enumerate}
        \item If $N(u) \cap N(w) = \emptyset$, no vertex of $N(u)$ has more than one neighbor in $N(w)$, so the edge between $N(u)$ and $N(w)$ form a matching.
        
        \item Suppose that $\{v\} = N(u) \cap N(w)$. 
        If there is some $a \in N(u)$ that $v$ is not adjacent to, $v$ is not adjacent to \textit{any} vertex $b \in N(u)$, otherwise we would have $\{a, u, b, v\}$ as an induced diamond of $G$.
        In turn, $v$ is either adjacent the entirety of $N(u)$ or to no vertex of $N(u)$ (the same holds with respect to $N(w)$).
        However, since $v$ is not universal and w.l.o.g, it is not adjacent to $N(u)$.
        We have two subcases.
        \begin{enumerate}
            \item If $v$ is not adjacent to $N(w)$, each vertex in $N(w)$ must be adjacent to every vertex in $N(u)$, otherwise we would have an induced $3P_1$; however, if $\max\left\{|N(u)|, |N(w)|\right\} \geq 2$, we have an induced diamond, so $G$ must be a $C_5$, a contradiction to our hypothesis.
            \item If $v$ is adjacent to $N(w)$ we are done: $V(G)$ is partitioned into the cliques $N[u], N[w]$ and the edges between them form a matching.
        \end{enumerate}
        \item Finally, if $\{v,v'\} = N(u) \cap N(w)$, we cannot have both $v,v'$ non-adjacent to the same vertex $a \in N(u) \setminus N(w)$, otherwise $\{a,v,v'\}$ form an induced $3P_1$ of $G$, nor have both of them adjacent to $a$, otherwise $\{v, a, u, v'\}$ would form an induced diamond; as such, suppose $v$ is adjacent to every vertex in $N(u) \setminus \{v, v'\}$.
        We again have to branch in two subcases.
        \begin{itemize}
            \item If $v$ is not adjacent to the vertices in $N(w) \setminus \{v,v'\}$, $v'$ must be adjacent to them.
            Therefore, we can partition $G$ into the cliques $(N[u] \setminus v'), (N[w]  \setminus v)$ and the edges between the cliques form a matching, otherwise we would have an induced diamond.
            \item If $v$ is adjacent to $N(w) \setminus \{v, v'\}$, then at least one of $N(w) \setminus \{v, v'\}$, $N(u) \setminus \{v, v'\}$ is empty, otherwise we have either a copy of $3P_1$ that includes $v'$ or an induced diamond $\{u,v,a,b\}$, where $a \in N(u)$ and $b \in N(w)$.
        \end{itemize}
    \end{enumerate}
    
    It is straightforward to verify that $\diam(G) \leq 3$.
    The converse follows directly and we omit it for brevity.
\end{proof}

The proof of the following theorem relies on the fact that the only minimal separators are either the universal cut vertex or we must choose, for each edge between the cliques, exactly one of its endpoints.
In the first case, the answer for \pname{Vertex Separator Reconfiguration} under TAR/TJ is always \YES, while under TS it suffices to check whether the number of tokens on each clique is equal in the initial and final separators; for an example of this case we refer to Figure~\ref{fig:3p1_diamond_case_1}.
For the latter, the case under TAR/TJ boils down to the same analysis.
For TS however, things require a bit more of work: if $u,v$ are the vertices we want to separate and every edge between $Q_1$ and $Q_2$ has an endpoint in $\{u,v\}$, then the analysis is also equivalent to the previous case; if, on the other hand, there is at least one edge that has neither $u$ nor $v$ as an endpoint, we can freely move tokens between $Q_1$ and $Q_2$ through it.
For an example of each of theses cases, we point to Figure~\ref{fig:3p1_diamond_case_2}.

\begin{theorem}
    Under TAR/TJ and TS, \pname{Vertex Separator Reconfiguration} can be solved in polynomial time for $\{3P_1, diamond\}$-free graphs. Furthermore, under TAR/TJ, it is always possible to reconfigure one separator into another.
\end{theorem}

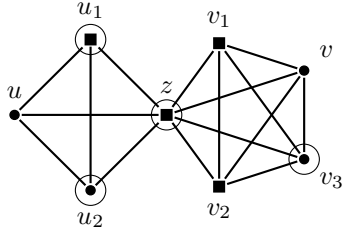
\begin{figure}[!htb]
    \centering
    \begin{tikzpicture}[scale=1]
            %\draw[help lines] (-5,-5) grid (5,5);
            \GraphInit[unit=3,vstyle=Normal]
            \SetVertexNormal[Shape=circle, FillColor=black, MinSize=2pt]
            \tikzset{VertexStyle/.append style = {inner sep = \inners, outer sep = \outers}}
            \SetVertexLabelOut
            \Vertex[x=-2, y = 1, Lpos=90, Math, L = {u_1}, Ldist=2pt]{u1}
            \Vertex[x=-2, y = -1, Lpos=270, Math, L = {u_2}, Ldist=2pt]{u2}
            \Vertex[x=-3, y = 0, Lpos=90, Math, L = {u}]{u}
            \Edges(u1,u2)
            
            \Vertex[x=-1, y = 0, Lpos=90, Math, L = {z}, Ldist=2pt]{z}
            
            \Vertex[a=108, d = 1, Lpos=90, Math, L = {v_1}]{v1}
            \Vertex[a=36, d = 1, Lpos=36, Math, L = {v}]{v}
            \Vertex[a=324, d = 1, Lpos=324, Math, L = {v_3}]{v3}
            \Vertex[a=252, d = 1, Lpos=270, Math, L = {v_2}]{v2}
            \begin{scope}
                \tikzset{VertexStyle/.append style = {shape = rectangle, inner sep = 2pt}}
                \Vertex[NoLabel, Node]{v1}
                \Vertex[NoLabel, Node]{v2}
                \Vertex[NoLabel, Node]{u1}
                \Vertex[NoLabel, Node]{z}
                \draw (-2, 1) circle (0.2cm);
                \draw (-2, -1) circle (0.2cm);
                \draw (-1, 0) circle (0.2cm);
                \draw (324:1) circle (0.2cm);
            \end{scope}
            
            \Edges(u, u1, z, u2, u, z, v1, v2, v3, v, z, v2, v, v1, v3, z)
    \end{tikzpicture}
    \caption{Under TAR/TJ, we can easily reconfigure the $uv$-separator $S_a = \{u_1, z, v_1, v_2\}$ into $S_b = \{u_1, u_2, z, v_3\}$, but not under TS since we cannot slide any token from the left clique to right without connecting $u$ and $v$.\label{fig:3p1_diamond_case_1}}
\end{figure}

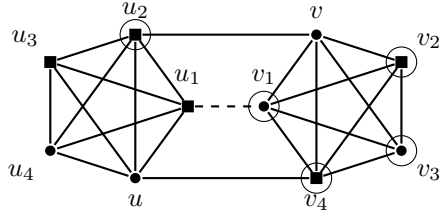
\begin{figure}[!htb]
    \centering
    \begin{tikzpicture}[scale=1]
            %\draw[help lines] (-5,-5) grid (5,5);
            \GraphInit[unit=3,vstyle=Normal]
            \SetVertexNormal[Shape=circle, FillColor=black, MinSize=2pt]
            \tikzset{VertexStyle/.append style = {inner sep = \inners, outer sep = \outers}}
            \SetVertexLabelOut
            
            \begin{scope}[xshift=-3cm, rotate=180]
                \Vertex[a=180, d = 1, Lpos=90, Math, L = {u_1}, Ldist=2pt]{u1}
                \Vertex[a=108, d = 1, Lpos=270, Math, L = {u}]{u}
                \Vertex[a=36, d = 1, Lpos=216, Math, L = {u_4}]{u4}
                \Vertex[a=324, d = 1, Lpos=108, Math, L = {u_3}]{u3}
                \Vertex[a=252, d = 1, Lpos=90, Math, L = {u_2}]{u2}
                \Edges(u,u1,u2,u3,u4,u,u2,u4,u1,u3,u)
                \draw (252:1) circle (0.2cm);
            \end{scope}

            \begin{scope}
                \Vertex[a=180, d = 1, Lpos=90, Math, L = {v_1}, Ldist=2pt]{v1}
                \Vertex[a=108, d = 1, Lpos=90, Math, L = {v}]{v}
                \Vertex[a=36, d = 1, Lpos=36, Math, L = {v_2}]{v2}
                \Vertex[a=324, d = 1, Lpos=324, Math, L = {v_3}]{v3}
                \Vertex[a=252, d = 1, Lpos=270, Math, L = {v_4}]{v4}
                \draw (324:1) circle (0.2cm);
                \draw (180:1) circle (0.2cm);
                \draw (252:1) circle (0.2cm);
                \draw (36:1) circle (0.2cm);
                \Edges(v,v1,v2,v3,v4,v,v2,v4,v1,v3,v)
            \end{scope}
            
            \Edges(u,v4)
            \Edges(u2,v)
            \Edge[style = {dashed}](u1)(v1)

            \begin{scope}
                \tikzset{VertexStyle/.append style = {shape = rectangle, inner sep = 2pt}}
                \Vertex[NoLabel, Node]{v4}
                \Vertex[NoLabel, Node]{u2}
                \Vertex[NoLabel, Node]{v2}
                \Vertex[NoLabel, Node]{u1}
                \Vertex[NoLabel, Node]{u3}
            \end{scope}
            
    \end{tikzpicture}
    \caption{If edge $u_1v_1 \in E(G)$, under all three rules, we can easily reconfigure the $uv$-separator $S_a = \{u_1, u_2, u_3, v_2, v_4\}$ into $S_b = \{u_2, v_1, v_2, v_3, v_4\}$; specifically, under TS, we can use edge $u_1v_1$ as passageway for the tokens on the left clique. If $u_1v_1 \notin E(G)$, we cannot reconfigure $S_a$ into $S_b$, since there is no way to move tokens from the left clique to the right.\label{fig:3p1_diamond_case_2}}
\end{figure}

\subsubsection{Series Parallel}

The goal of this section is to show that, under TJ, it is always possible to reconfigure two $s,t$ separators in polynomial time.
\textit{Series-parallel} graphs have multiple characterizations, such as being exactly the graphs of treewidth 2~\cite{2_trees}.
One that is particularly useful for us is the recursive definition of \citeauthor{series_parallel}~\cite{series_parallel}, which we adapt below.

\begin{definition}
    Let $G$ be a multigraph.
    \begin{enumerate}
        \item If $|V(G)| = 1$ and $G$ has a loop, $G$ is a series-parallel graph;
        \item If $G$ is series-parallel, the graph $H$ obtained by subdividing an edge of $G$ is series-parallel.
        \item If $G$ is series-parallel, the graph $H$ obtained by removing one edge $e = uv$ of $G$ and replacing by two edges $f = g = e$ is series parallel.
    \end{enumerate}
    Operation 2 is known as the \textit{series} operation (or S), and operation 3 as the \textit{parallel} operation (or P).
    $G$ is series-parallel if and only if its 2-connected components can be obtained from a loop by repeatedly applying operations S and P.
    Observation~\ref{obs:sp_nontame} follows directly from this recursive definition.
\end{definition}

\begin{observation}
    \label{obs:sp_nontame}
    The class of series-parallel graphs is not tame.
\end{observation}

Note that, if we are concerned with 2-connected graphs, the first operation that acts on the loop must be a series operation, otherwise $G$ would have a cut vertex; since this is the case, we may, equivalently, assume that, instead of a loop, our initial graph is an edge and the first operation must be of type $P$.
Let $\angled{G_1 , G_2, \dots, G_k}$ be a sequence of graphs such that $G_1$ is an edge, $G_k = G$ and for all $i \in [k-1]$, there is some edge of $G_i$ upon which we apply either operation $S$ or $P$ to obtain $G_{i+1}$, and let $\theta = \angled{(\phi_1, e_1), (\phi_2, e_2), \dots, (\phi_{k-1}, e_{k-1})}$ be the corresponding sequence of $S,P$ operations; i.e. $\phi_i \in \{S,P\}$ and $e_i \in E(G_i)$ but $e_i \notin E(G_j)$ for any $j > i$.
By our definition, $(\phi_1, e_1) = (P, e_1)$.

We say that edge $e \in G_j$ is a \textit{descendant} of edge $e_i \notin \bigcup_{\ell \geq j} E(G_\ell)$ if $e_j$ is the result of $\phi_i$ applied to $e_i$ or if it is the result of $\phi_r$, $r < j$ applied to some descendant $e_r$ of $e_i$.
Two edges $e, f$ are \textit{unrelated} if $e$ is not a descendant of $f$ and vice-versa.
In particular, every edge of $G$ is a descendant of the original loop and, if $e$ is a descendant of $e_j$ and of $e_\ell$, either $e_j$ is a descendant of $e_\ell$ or vice-versa.
This relationship can be seen as a full binary rooted tree $\mathcal{T}$, where each internal vertex $t_i$ corresponds to some $(\phi_i, e_i)$, the leaves are precisely the edges of $G$, and $t \in V(\mathcal{T})$ is a descendant of $t_i$ in the tree if and only if the edge associated to $t$ is a descendant of $e_i$.
Moreover, the root of the tree is given by $(\phi_1, e_1)$ with $\phi_1 = P$.
We call $\mathcal{T}$ a \textit{PS-tree} of $G$ and present an example of one in Figure~\ref{fig:ps_tree}.

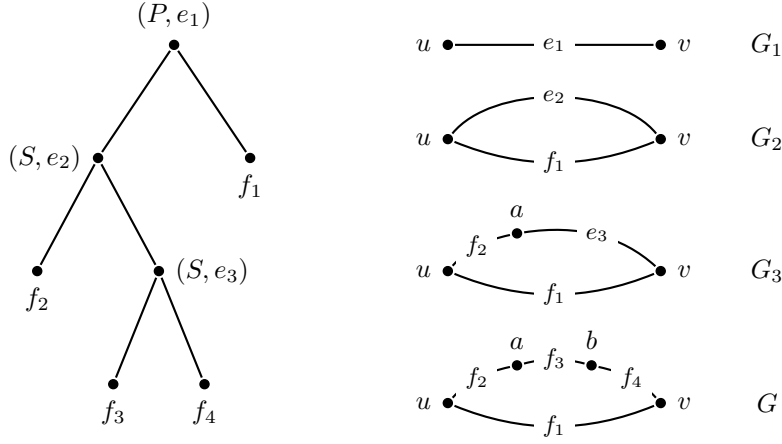
\begin{figure}[!htb]
    \centering
    \begin{tikzpicture}[scale=1]
            %\draw[help lines] (-5,-5) grid (5,5);
            \GraphInit[unit=3,vstyle=Normal]
            \SetVertexNormal[Shape=circle, FillColor=black, MinSize=2pt]
            \tikzset{VertexStyle/.append style = {inner sep = \inners, outer sep = \outers, font=\small}}
            \tikzset{EdgeStyle/.append style = {font = \small}}
            \SetVertexLabelOut
            
            \begin{scope}[xshift=-2cm]
                \Vertex[x=0, y=0, Lpos=90, Math, L={(P, e_1)}]{e1}
                \Vertex[x=1, y=-1.5, Lpos=270, Math, L={f_1}]{f1}
                \Vertex[x=-1, y=-1.5, Lpos=180, Math, L={(S, e_2)}]{e2}
                \Vertex[x=-1.8, y=-3, Lpos=270, Math, L={f_2}]{f2}
                \Vertex[x=-0.2, y=-3, Lpos=0, Math, L={(S, e_3)}]{e3}
                \Vertex[x=-0.8, y=-4.5, Lpos=270, Math, L={f_3}]{f3}
                \Vertex[x=0.4, y=-4.5, Lpos=270, Math, L={f_4}]{f4}
                \Edges(f1, e1, e2, f2)
                \Edges(e2, e3, f3)
                \Edges(f4, e3)
            \end{scope}
            
            \begin{scope}[xshift=3cm,xscale=0.7, yscale=0.5]
                \begin{scope}
                    \node at (4,0) {$G_1$};
                    \Vertex[x=-2, y = 0, Lpos=180, Math, L={u}]{u1}
                    \Vertex[x=2, y = 0, Lpos=0, Math, L={v}]{v1}
                    \Edge[label={$e_1$}](u1)(v1)
                \end{scope}
                
                \begin{scope}[yshift=-2.5cm]
                    \node at (4,0) {$G_2$};
                    \Vertex[x=-2, y = 0, Lpos=180, Math, L={u}]{u2}
                    \Vertex[x=2, y = 0, Lpos=0, Math, L={v}]{v2}
                    \Edge[style={bend left=60}, label={$e_2$}](u2)(v2)
                    \Edge[style={bend right}, label={$f_1$}](u2)(v2)
                \end{scope}

                \begin{scope}[yshift=-6cm]
                    \node at (4,0) {$G_3$};
                    \Vertex[x=-2, y = 0, Lpos=180, Math, L={u}]{u3}
                    \Vertex[x=2, y = 0, Lpos=0, Math, L={v}]{v3}
                    \Vertex[x=-0.7, y=1, Lpos=90, Math, L={a}]{a3}
                    
                    \Edge[style={bend left=20}, label={$f_2$}](u3)(a3)
                    \Edge[style={bend left=32}, label={$e_3$}](a3)(v3)
                    
                    \Edge[style={bend right}, label={$f_1$}](u3)(v3)
                \end{scope}
                
                \begin{scope}[yshift=-9.5cm]
                    \node at (4,0) {$G$};
                    \Vertex[x=-2, y = 0, Lpos=180, Math, L={u}]{u4}
                    \Vertex[x=2, y = 0, Lpos=0, Math, L={v}]{v4}
                    \Vertex[x=-0.7, y=1, Lpos=90, Math, L={a}]{a4}
                    \Vertex[x=0.7, y=1, Lpos=90, Math, L={b}]{b4}
                    
                    \Edge[style={bend left=20}, label={$f_2$}](u4)(a4)
                    \Edge[style={bend left=20}, label={$f_3$}](a4)(b4)
                    \Edge[style={bend left=20}, label={$f_4$}](b4)(v4)
                    
                    \Edge[style={bend right}, label={$f_1$}](u4)(v4)
                \end{scope}
            \end{scope}

            %\draw (-1.7,1.3) rectangle (1.7,1.7);
    \end{tikzpicture}
    \caption{A PS-tree and the corresponding sequence of series-parallel multigraphs. \label{fig:ps_tree}}
\end{figure}

We say that a vertex $v$ is a \textit{piece} of edge $e_i$ if $v \notin V(G_i)$, $(S, e_i)$ is in $\theta$ and $v$ is an endpoint of one of the descendants of $e_i$; the unique edge $e_i$ that was subdivided to create $v \notin V(G_1)$ is the \textit{support} of $v$, which we denote by $\supp(v) = e_i$.
The \textit{span} $\espan(v) = \angled{f_1, \dots, f_p}$ of $v$ is the maximum sequence of edges of $\bigcup_{i \in [k]} E(G_k)$ that $v$ is a piece of, such that $f_p = \supp(v)$ and $f_j$ is a descendant of $f_i$ if $i < j$.
Note that, for each $f_i \in \espan(v)$, removing the endpoints of $f_i$ from $G$ separates $v$ from all vertices that are not pieces of $f_i$; moreover, by  definition of piece, one vertex cannot be a piece of an edge that was the target of a parallel operation and, consequently, $f_i \cap f_{i+1}$ has precisely one vertex.
This implies that if a separator of $G$ contains the endpoints of $f_i$ but does not contain one endpoint of $f_{i+1}$, we can reconfigure it to contain the latter with one jump operation.
Two vertices $u,v \in V(G)$ are a \textit{parallel pair} if the support $e$ of $u$ and the support $f$ of $v$ are unrelated, and the lowest common ancestor of  $e,f$ in $\mathcal{T}$ is $(P, e_i = ab)$; in this case, we say that $u,v$ are \textit{parallel with respect to} $a,b$.
Now, suppose that the lowest common ancestor of $e = \supp(u)$ and $f = \supp(v)$ is an $S$ node and $e,f$ are unrelated. If there is some $G_i$ where $uv \in E(G_i)$ but $uv \notin E(G)$, $u,v$ are a \textit{sequential pair}; if no $G_i$ has $uv \in E(G_i)$, $u,v$ form a \textit{serial pair}; the third case, that has $uv \in E(G)$ is not relevant to us since we are concerned with vertex separators and we cannot separate adjacent vertices.
Vertices $u,v$ form a \textit{nested pair} if no $G_i$ has $uv \in E(G_i)$ and $\supp(v)$ is a descendant of an edge that has $u$ as an endpoint. %, i.e the lowest common ancestor of $\supp(u)$ and $\supp(v)$ is $\supp(u)$. 
We denote by $\varepsilon(u,v)$ the number of times some edge with $u$ and $v$ as endpoints was the target of a $P$ operation in $(G_1, \dots, G_k)$, if there is no $G_i$ that contains $uv$, we define $\varepsilon(u,v) = 0$.

A multigraph is series-parallel if each of its two-connected components is series parallel.
A graph is series-parallel if it can be obtained from a series-parallel multigraph by removing all but one edge of each set of parallel edges.
Our next goal is to show that, under TJ, it is always possible to reconfigure any two $st$-separators of a series-parallel graph $G$.
We focus on vertices belonging to a same two-connected component, otherwise there is a cut vertex between them and it suffices to place one token on it in the first intermediate step.
Thus, for the remainder of this section, we assume that $G$ is a two-connected series-parallel graph.

\begin{observation}
    \label{obs:edge_sep}
    For any non-adjacent $s,t \in V(G)$ such that $s$ is a piece of edge $e_i = ab$ and $t$ is not a piece nor an endpoint of $e_i$, $\{a,b\}$ is an $st$-separator.
\end{observation}

\begin{lemma}
    \label{lem:canon_sizes}
    Let $s,t \in V(G)$ be two nonadjacent vertices. If $\{s,t\} = V(G_1)$, the size of a minimum $st$-separator is equal to $\varepsilon(s,t) + 1$, otherwise it is equal to $\varepsilon(s,t) + 2$.
\end{lemma}

\begin{proof}
    Let $V(G_1) = \{u,v\}$.
    We begin our analysis on the relationship between $s,t$, first supposing $\{s,t\} \cap \{u,v\} = \emptyset$.
    \begin{itemize}
        \item[Serial Pair] Let $(S, e_i) \in V(\mathcal{T})$ be the lowest common ancestor of the supports of $s$ and $t$, and $(\phi_\ell, e_\ell), (\phi_r, e_r)$ its children, such that $\supp(t)$ is $e_\ell$ or a descendant of it, $\supp(s)$ is $e_r$ or a descendant of it, $e_i = ab$, $e_\ell = az$, and $e_r = zb$.
        To see that $\{a,z\}$ is a minimum $st$-separator, note that there can be no piece of $e_\ell$ adjacent to a piece of $e_r$, so any $s-t$ path that uses a piece of $e_r$ must include either $z$ and any $s-t$ path that does not pass through $z$ necessarily contains $a$.
        Since $\varepsilon(s,t) = 0$, the statement holds.
        An example of a serial pair $s,t$ is shown in Figure~\ref{fig:serial_pair}.
        
        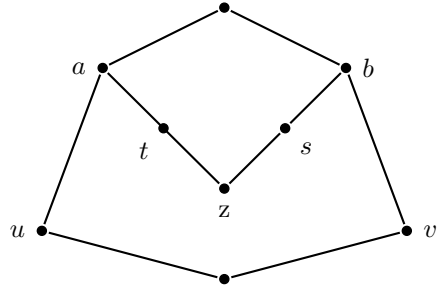
\begin{figure}[!htb]
            \centering
            \begin{tikzpicture}[scale=0.8]
                    %\draw[help lines] (-5,-5) grid (5,5);
                    \GraphInit[unit=3,vstyle=Normal]
                    \SetVertexNormal[Shape=circle, FillColor=black, MinSize=2pt]
                    \tikzset{VertexStyle/.append style = {inner sep = \inners, outer sep = \outers}}
                    \SetVertexLabelOut
                    \Vertex[x=-3, y = 0.3, Lpos=180, Math]{u}
                    \Vertex[x=3, y = 0.3, Lpos=0, Math]{v}
                    
                    \Vertex[x=-2, y = 3, Lpos=180, Math]{a}
                    \Vertex[x=2, y = 3, Lpos=0, Math]{b}
                    
                    \Vertex[x=1, y = 2, Math, Lpos= 315]{s}

                    \Vertex[x=0, y = 4, NoLabel]{a_1}
                    \Vertex[x=0, y = 1, Lpos=270]{z}
                    \Vertex[x=0, y = -0.5, NoLabel]{a_2}
                    
                    \Edge(a)(z)
                    
                    \Vertex[x=-1, y = 2, Math, Lpos=225]{t}
                    
                    \Edges(z, s, b, v, a_2, u, a, a_1, b)

                    %\draw (-1.7,1.3) rectangle (1.7,1.7);
            \end{tikzpicture}
            \caption{A serial pair $s,t$ with $e_i = ab$, $\supp(z) = e_i$ and $M(s,t) = \{a,z\}$. \label{fig:serial_pair}}
        \end{figure}

        \item[Nested Pair] Let $(S, e_i)$ be the node of $\mathcal{T}$ where, without loss of generality, $e_i = as$, $(\phi_\ell, e_\ell), (\phi_r, e_r)$ its children, such that $\supp(t)$ is $e_\ell = az$ or a descendant of $e_\ell$, and $e_r = zs$.
        By Observation~\ref{obs:edge_sep}, $\{a,z\}$ is an $st$-separator and is minimum because $G$ is 2-connected.
        Since $\varepsilon(s,t) = 0$, the statement holds.
        Figure~\ref{fig:nested_pair} gives an example of a nested pair $s,t$.

        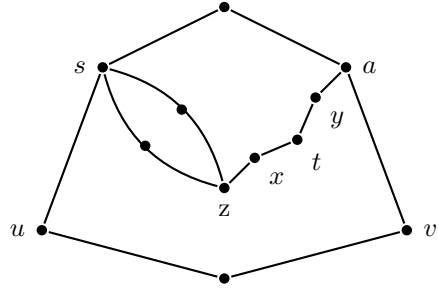
\begin{figure}[!htb]
            \centering
            \begin{tikzpicture}[scale=0.8]
                    %\draw[help lines] (-5,-5) grid (5,5);
                    \GraphInit[unit=3,vstyle=Normal]
                    \SetVertexNormal[Shape=circle, FillColor=black, MinSize=2pt]
                    \tikzset{VertexStyle/.append style = {inner sep = \inners, outer sep = \outers}}
                    \SetVertexLabelOut
                    \Vertex[x=-3, y = 0.3, Lpos=180, Math]{u}
                    \Vertex[x=3, y = 0.3, Lpos=0, Math]{v}
                    
                    \Vertex[x=-2, y = 3, Lpos=180, Math]{s}
                    \Vertex[x=2, y = 3, Lpos=0, Math]{a}
                    
                    \Vertex[x=1.2, y = 1.8, Lpos=315, Math]{t}
                    \Vertex[x=0.5, y = 1.5, Lpos=315, Math]{x}
                    \Vertex[x=1.5, y = 2.5, Lpos=315, Math]{y}

                    \Vertex[x=0, y = 4, NoLabel]{a_1}
                    \Vertex[x=0, y = 1, Lpos=270]{z}
                    \Vertex[x=0, y = -0.5, NoLabel]{a_2}
                    
                    \Edge[style={bend right}](s)(z)
                    \Edge[style={bend left}](s)(z)
                    
                    \Vertex[x=-0.7, y = 2.3, NoLabel]{z_2}
                    \Vertex[x=-1.3, y = 1.7, NoLabel]{z_3}
                    
                    \Edges(z, x, t, y, a, v, a_2, u, s, a_1, a)
                    
                    %\draw (-1.7,1.3) rectangle (1.7,1.7);
            \end{tikzpicture}
            \caption{A nested pair $s,t$ with $\supp(t)$ being an $xy$-edge and $M(s,t) = \{a,z\}$ \label{fig:nested_pair}}
        \end{figure}
        
        \item[Sequential Pair] Suppose without loss of generality that $\supp(t) = e_i = sa$, define $V(s,t)$ as the vertices $z_j$ that have an $st$-edge as support,
        and let $f_j = sz_j$ and $g_j = z_jt$ be the edges created by $(S, e_j) \in V(\mathcal{T})$.
        We show that $\{a\} \cup V(s,t)$ is an $st$-separator of minimum size by induction on $|V(s,t)|$.
        For the base case of $V(s,t) = \{z_j\}$, no piece of $f_j$ is adjacent to a piece of $g_j$, so any $s - t$ path that uses a piece of either of these edges must pass through $z_j$; similarly, any path that does not use a piece of $st$ must pass through $a$.
        By Menger's Theorem, $\{z_j, a\}$ is a minimum $st$-separator.
        For the general case, take any $z_j \in V(s,t)$, let $G'$ be the induced subgraph of $G$ obtained by the removal of all pieces of $\supp(z_j)$, and note that $G'$ is a series parallel graph with one fewer operation $P$ applied to an edge $st$.
        By induction, $V(s,t) \setminus \{z_j\} \cup \{a\}$ is an $st$-separator of $G'$.
        Since the only $s-t$ paths on $G$ not on $G'$ are those that use exclusively pieces of $\supp(z_j)$ and must pass through $z_j$, we have that $V(s,t) \cup \{a\}$ is an $st$-separator.
        Note that $|V(s,t)| = \varepsilon(s,t) + 1$ since each operation $P$ increases $\varepsilon(s,t)$ by one and the first $st$-edge is the result of an $S$ operation.
        As such, it holds that the size of a minimum $s-t$ separator is $|\{a\} \cup V(s,t)| = \varepsilon(s,t) + 2$.
        For an example of a sequential pair $s,t$ and the set $V(s,t)$, we refer to Figure~\ref{fig:sequential_pair}.
        
        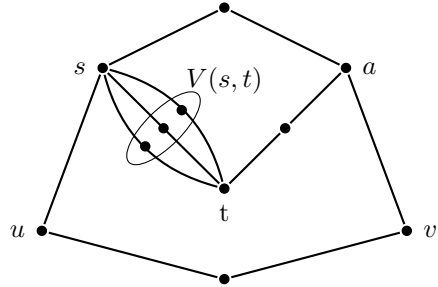
\begin{figure}[!htb]
            \centering
            \begin{tikzpicture}[scale=0.8]
                    %\draw[help lines] (-5,-5) grid (5,5);
                    \GraphInit[unit=3,vstyle=Normal]
                    \SetVertexNormal[Shape=circle, FillColor=black, MinSize=2pt]
                    \tikzset{VertexStyle/.append style = {inner sep = \inners, outer sep = \outers}}
                    \SetVertexLabelOut
                    \Vertex[x=-3, y = 0.3, Lpos=180, Math]{u}
                    \Vertex[x=3, y = 0.3, Lpos=0, Math]{v}
                    
                    \Vertex[x=-2, y = 3, Lpos=180, Math]{s}
                    \Vertex[x=2, y = 3, Lpos=0, Math]{a}
                    
                    \Vertex[x=1, y = 2, NoLabel]{y_2}

                    \Vertex[x=0, y = 4, NoLabel]{a_1}
                    \Vertex[x=0, y = 1, Lpos=270]{t}
                    \Vertex[x=0, y = -0.5, NoLabel]{a_2}
                    
                    \Edge[style={bend right}](s)(t)
                    \Edge[](s)(t)
                    \Edge[style={bend left}](s)(t)
                    
                    \Vertex[x=-1, y = 2, NoLabel]{z_1}
                    \Vertex[x=-0.7, y = 2.3, NoLabel]{z_2}
                    \Vertex[x=-1.3, y = 1.7, NoLabel]{z_3}
                    
                    \Edges(t, y_2, a, v, a_2, u, s, a_1, a)
                    
                    \begin{scope}[xshift=-1cm, yshift=2cm, rotate=45]
                        \draw (0,0) circle [x radius=0.8cm, y radius=0.3cm];
                    \end{scope}
                    \node at (0, 2.8) {$V(s,t)$};
                    
                    %\draw (-1.7,1.3) rectangle (1.7,1.7);
            \end{tikzpicture}
            \caption{A sequential pair $s,t$ and the set $V(s,t)$, with $\supp(t)$ being an $sa$-edge. \label{fig:sequential_pair}}
        \end{figure}
        
        \item[Parallel Pair] If $(P, e_i)$ is the node of $\mathcal{T}$ corresponding to the lowest common ancestor of the nodes containing $\supp(s)$ and $\supp(t)$, with $e_i = ab$, then $\{a,b\}$ is a minimum $st$-separator: any $s-t$ path must include either $a$ or $b$.
        Since $\varepsilon(s,t) = 0$, the statement holds.
        Figure~\ref{fig:parallel_pair} presents an example of an $s,t$ parallel pair which is parallel relative to $a,b$.
        
        \begin{figure}[!htb]
            \centering
            \begin{tikzpicture}[scale=0.8]
                    %\draw[help lines] (-5,-5) grid (5,5);
                    \GraphInit[unit=3,vstyle=Normal]
                    \SetVertexNormal[Shape=circle, FillColor=black, MinSize=2pt]
                    \tikzset{VertexStyle/.append style = {inner sep = \inners, outer sep = \outers}}
                    \SetVertexLabelOut
                    \Vertex[x=-3, y = 0, Lpos=180, Math]{a}
                    \Vertex[x=3, y = 0, Lpos=0, Math]{b}
                    
                    \Vertex[x=-2, y = 3, Lpos=180, Math]{x_1}
                    \Vertex[x=2, y = 3, Lpos=0, Math]{y_1}
                    
                    \Vertex[x=-1, y = 1.5, Lpos=180, Math]{x_2}
                    \Vertex[x=1, y = 1.5, Lpos=0, Math]{y_2}

                    \Vertex[x=0, y = 4, NoLabel]{z_1}
                    \Vertex[x=0, y = 2.5, Lpos=90, Math]{t}
                    \Vertex[x=0, y = 0.5, NoLabel]{z_2}
                    \Vertex[x=0, y = -1, Lpos=270, Math]{s}
                    
                    \Edges(x_1,a,s,b,y_1,z_1,x_1,x_2,z_2,y_2,t,x_2)
                    \Edges(y_2,y_1)
                    
                    \draw (-1.7,1.3) rectangle (1.7,1.7);
                    \draw (0,0) circle [x radius=3.7cm, y radius=0.4cm];
            \end{tikzpicture}
            \caption{To reconfigure $A = \{x_2,y_2\}$ into $M(s,t) = \{a, b\}$, we need to swap $x_2$ with $x_1$, $y_2$ with $y_1$, $x_1$ with $a$, and finally $y_1$ with $b$.\label{fig:parallel_pair}}
        \end{figure}
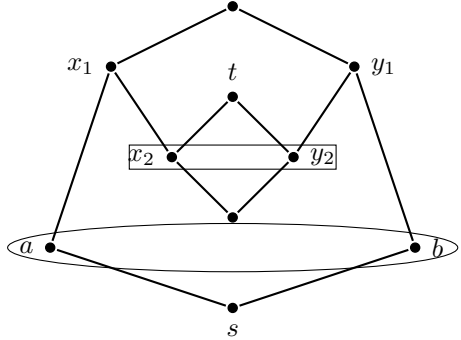
        
    \end{itemize}
    
    Now, suppose $s = u$.
    If $t = v$ or there is some $e_i = st$, a similar argument to the one used for sequential pairs holds: in the first case, $V(s,t)$ is a minimum separator of size $\varepsilon(s,t) + 1$, in the second, $V(s,t) \cup (\supp(t) \setminus \{s\})$ is a minimum separator of size $\varepsilon(s,t) + 2$.
    Otherwise, if there is no $e_i = st$, let $z$ be the unique vertex such that $\supp(z)$ is an $sa$-edge and $\supp(t)$ is either $za$ or a descendant of $za$.
    In this case, $\{z,a\}$ is a minimum $st$-separator: every $s-t$ path passes either through $z$ or $a$.
\end{proof}

We refer to the separators described in the proof of our previous theorem as the \textit{canonical $st$-separators}, and denote them by $M(s,t)$.

From now on, our strategy for reconfiguring an $s-t$ separator $A$ into another $s-t$ separator $B$ will consist in reconfiguring them into some intermediate separator containing $M(s,t)$. 
Note that this may be far from trivial; for instance, when $s,t \in V(G)$ form a parallel pair, it may require quite a bit of work to reconfigure a separator to contain $M(s,t)$, as it may not be possible to directly exchange a token in $A$ for a token in $M(s,t)$, as shown by the example in Figure~\ref{fig:parallel_pair}.

\begin{lemma}
    \label{lem:parallel_pathway}
    For any $s,t \in V(G)$ and minimal $st$-separator $A$, if there is some $w \in A$ such that $w,t$ form a parallel pair with respect to $x,y$, $x,y \notin A$, and $s$ is not a piece of the same $xy$-edge as $w$, then, on $G \setminus A$, $x$ is reachable from $t$, $y$ is not, and, on $G \setminus (A \setminus \{w\})$, every $t-s$ path contains $x$, $w$, and $y$, in this order.
\end{lemma}

\begin{proof}
    Suppose $x$ is reachable from $t$ on $G \setminus A$; if none were reachable, then there would not be any need for $w$, since any $t-w$ path contains at least one of $x,y$, contradicting the minimality of $A$.
    On the other hand, if both $x,y$ were reachable, then $A \setminus \{w\}$ would also be an $st$-separator since $s$ is not a piece of the same $xy$-edge as $w$, any $s-t$ path that passes through $w$ must pass through \emph{both} $x$ and $y$.
    As such, only $x$ is reachable from $t$.
    Now, for the second statement, there is no $x-s$ path on $G \setminus A$ but, because $A$ is minimal, every $s-t$ path on $G \setminus (A \setminus \{w\})$ must contain $w$ and, since $s$ is not a piece of the same $xy$-edge as $w$, every $s-t$ path on $G \setminus (A \setminus \{w\})$ is of the form $t-x-w-y-s$.
\end{proof}

\begin{lemma}
    \label{lem:span_separator}
    For any $s,t \in V(G)$, let $e_j = az$ be any edge such that $t$ is a piece of $e_j$ but $s$ is not a piece nor an endpoint of $e_j$, and $F = \angled{f_\ell, \dots, f_p}$ be the suffix of $\espan(t)$ where $f_\ell = e_j$.
    If $A$ is a minimal $st$-separator contained in the set of pieces of $e_j$, then $A$ can be reconfigured to a separator that contains one vertex that is an endpoint of some $f_i \in F$ and, consequently, $A$ can be reconfigured to a separator that contains $\{a,z\}$.
\end{lemma}

\begin{proof}
    Let us first show that if $A$ contains one endpoint of edge $f_i \in F$, then we can reconfigure $A$ into a separator that contains $a,z$.
    If this is the case, we can move any other token to the other endpoint of $f_i$; by Observation~\ref{obs:edge_sep}, the new set is still an $st$-separator.
    At this point, since $f_i \cap f_{i-1} = x_i$ for all $i$, we can iteratively move the token in $f_{i-1} \setminus \{x_i\}$ to $f_i \setminus \{x_i\}$ until we obtain $f_\ell = az$.
    For the remainder of this proof, we assume that $A$ does not intersect any of the endpoints of the edges in $F$.
    
    If $\supp(t) = \{x,y\}$ and there is some $w \in A$ that is a piece of an $xt$-edge, simply move the token on $w$ to $x$: the resulting set is still a separator since every $s-t$ path that contains $w$ necessarily contains $x$.
    We may now assume that $A$ contains no piece of the support of $t$.
    
    Suppose now that there is some $w \in A$ such that $t,w$ are parallel with respect to $x,y$ and let $f_i \in F$ be an $xy$-edge.
    By Lemma~\ref{lem:parallel_pathway}, which has its conditions satisfied, every $s-t$ path in $G \setminus (A \setminus \{w\})$ is of the form $t-x-w-y-s$ and we conclude that $A \cup \{x\} \setminus \{w\}$ is an $st$-separator.
    
    If there is no vertex of $A$ forming a parallel pair with $t$, let $f_i = xy$ be the edge of largest index in $F$ such that $x$ is unreachable from $t$ but $y$ is; note that there must be one such edge since $\supp(t)$ is reachable but neither $a$ nor $z$ are and $f_i \cap f_{i+1} \neq \emptyset$.
    By this argument, it must be that $f_{i+1} = wy$, that $w$ is reachable from $t$ on $G \setminus A$ and that edge $xw$ has at least one piece, otherwise $x$ would be reachable.
    Furthermore, there is no $xy \in E(G)$ nor any other $xy$-edge in $(G_1, \dots, G_k)$: if there was such an edge it would have to be the target of an $S$ operation, but in this case pieces of this edge would have been required to block any $y-x$ path in $G \setminus A$, thus there would be some piece forming a parallel pair with $t$.
    As such, for $x$ to be unreachable, there must be at least one piece of $xw$ in $A$.
    Moving this piece to $x$ does not create any $s-t$ path and the new separator now contains a vertex of $f_i$.
\end{proof}

\begin{lemma}
    \label{lem:parallel_pair}
    If $s,t$ form a parallel pair on $G$, then, under TJ, any minimal $st$-separator $A$ can be reconfigured into a separator that contains $M(s,t)$.
\end{lemma}

\begin{proof}
    Let $s,t$ be a parallel pair with respect to $a,b$ such that $s$ is a piece of $e_s = ab$ and $t$ a piece of $e_t = ab$.
    At first, $A$ can intersect pieces of any edge of the graph, but we can partition $A$ in $\{A_s, A_t, A^*\}$, where $A_s$ are the pieces of $e_s$ in $A$, $A_t$ the pieces of $e_t$, and $A^* = A \setminus A_s \setminus A_t$.
        
    If $A = A_t$ or $A = A_s$, i.e. $A$ is a subset of the pieces of $e_t$ or $e_s$, we can apply Lemma~\ref{lem:span_separator} and, noting that, in both cases the suffix of both spans begin at an $ab$-edge, we are done.
    On the other hand, if $A^* \neq \emptyset$, then there must be, say, $s-a$ paths on $G \setminus A$, but no $s-b$ path, and $t-b$ paths but no $t-a$ path on $G \setminus A$, implying that $A_s, A_t \neq \emptyset$.
    In this case, we can move any token from $A_s$ to $b$ without creating any $s-t$ path, then move any other token to $a$ to obtain $M(s,t)$.
    The case $A^* = \emptyset$ but $A_s, A_t \neq \emptyset$ is only possible if the only $a-b$ paths on $G$ pass through $e_s$ or $e_t$; i.e., $a,b \in V(G_1)$.
    The argument, however, is the same as in the previous case: $A$ destroys every $s-b$ and $t-a$ paths, but since the only $s-t$ paths are precisely $s-a-t$ and $s-b-t$, moving a token in $A_s$ to $b$ and then another to $a$ suffices to construct $M(s,t)$.
\end{proof}

\begin{lemma}
    \label{lem:serial_pair}
    If $s,t$ form a serial pair on $G$, then, under TJ, any minimal $st$-separator $A$ can be reconfigured into a separator that contains $M(s,t)$.
\end{lemma}

\begin{proof}
    Let $\{a,z\} = M(s,t)$ such that $e_i = ab$, $z \in V(G)$, $\supp(z) = e_i$, $\supp(t)$ is either $az$ or a descendant of $az$, $\supp(s)$ is either $zb$ or a descendant of it, and $F = \angled{f_\ell, \dots, f_p}$ be the suffix of $\espan(t)$ beginning at $f_\ell = az$; moreover, $A$ does not intersect any $f_i$ of the suffix, otherwise we could reconfigure it to $\{a, z\}$ as done in Lemma~\ref{lem:span_separator}, and $A$ does not contain a piece of an edge that is a descendant of $\supp(t)$, otherwise we could move this piece to the suitable endpoint of $\supp(t)$.

    Suppose that neither $a$ nor $z$ are reachable from $t$ on $G \setminus A$, implying that $A$ is a subset of the pieces of $f_\ell$.
    Note that, because of our last observation on the previous paragraph, neither $a$ nor $z$ are endpoints of $\supp(t)$.
    By Lemma~\ref{lem:span_separator}, we can reconfigure $A$ to contain $\{a,z\}$.
    
    If exactly one of $\{a,z\}$ is reachable from $t$, there is at least one piece of $f_\ell = az$ in $A$ that destroys the paths to one of the endpoints of $f_\ell$, \textit{but not to the other}. Thus, we can freely move any piece of $az$ to the unreachable endpoint of $f_\ell$ without creating any $s-t$ paths.
    
    The only case left, is if both $a$ and $z$ are reachable from $t$.
    If $b$ is not reachable, then $A$ contains a piece $w$ of edges other than $e_i$; thus, if we pick any of these pieces and move to $a$, vertex $b$ remains unreachable and so does $s$: any $z-b$ path that does not use pieces of $az$ is internally vertex disjoint with any $a-b$ path that does not uses pieces of $az$, and moving a piece of such an $a-b$ path does not create a $z-b$ path on $G \setminus (A \cup \{a\} \setminus \{w\})$.
    If, on the other hand, $b$ is reachable, then $A$ is a subset of the pieces of $zb$, and, moreover, neither $b$ nor $z$ are reachable from $s$ on $G \setminus A$.
    By the second paragraph, we can reconfigure $A$ to contain $\{b,z\}$, but in this case the reconfigured separator would intersect $F$, and we can move the token on $b$ to $a$ and obtain $M(s,t)$.
\end{proof}

\begin{lemma}
    \label{lem:sequential_pair}
    If $s,t$ form a sequential pair on $G$, then, under TJ, any minimal $st$-separator $A$ can be reconfigured into a separator that contains $M(s,t)$.
\end{lemma}

\begin{proof}
    Let $e_i = sa$ be the support of $t$ and $V(s,t)$ be the set of vertices $z_j$ that have an $st$-edge as their support.
    Let $A_s$ be the set of pieces of $st$-edges in $A$, and $A_a$ the set of pieces of $ta$-edges in $A$.
    Since $s$ cannot be adjacent to any piece $w$ of $ta$, any $s-w$ path between either contains $t$ or $a$; since no $s-t$ paths exist in $G \setminus A$, if $w \in A_a$, we can move the token on $w$ to $a$ without creating any $s-t$ path.
    Now, for each edge $e_j = st$, since the $s-t$ paths that contain a piece of $e_j$ are internally vertex disjoint, there must be at least one $w_j \in A_s$, but we can move the token on $w_j$ to $z_j$ without creating any $s-t$ path since no piece of $sz_j$ can be adjacent to a piece of $z_jt$.
    After placing tokens on all $z_j$'s, any other token can be moved to $a$, thus obtaining $M(s,t)$.
\end{proof}

\begin{lemma}
    \label{lem:nested_pair}
    If $s,t$ form a nested pair on $G$, then, under TJ, any minimal $st$-separator $A$ can be reconfigured into a separator that contains $M(s,t)$.
\end{lemma}

\begin{proof}
    Let $\{a,z\} = M(s,t)$ such that $e_i = as$ and $z$ are as in the proof of Lemma~\ref{lem:canon_sizes}, and $F = \angled{f_\ell, \dots, f_p}$ be the suffix of $\espan(t)$ beginning at $f_\ell = az$; moreover, $A$ does not intersect any $f_i$ of the suffix, otherwise we could reconfigure it to $\{a, z\}$ as done in Lemma~\ref{lem:span_separator}, and $A$ does not contain a piece of an edge that is a descendant of $\supp(t)$, otherwise we could move this piece to the suitable endpoint of $\supp(t)$.
    
    As was done in Lemma~\ref{lem:serial_pair}, if neither $a$ nor $z$ are reachable from $t$ on $G \setminus A$, $A$ is a subset of the pieces of $f_\ell$ and neither $a$ nor $z$ are endpoints of $\supp(t)$.
    By Lemma~\ref{lem:span_separator}, we can reconfigure $A$ to contain $\{a,z\}$.
    
    If $z$ is reachable from $t$ on $G \setminus A$, there necessarily is some piece $w$ of $zs$ in $A$, otherwise we would have a $t-z-s$ path; in this case, we can move the token on $w$ to $z$ and then any other token to $a$.
    Otherwise, if only $a$ is reachable from $t$, there is at least one piece of $f_\ell = az$ in $A$ that destroys the paths to $z$, \textit{but not to $a$}.
    Thus, we can move the token on any piece of $az$ in $A$ to $z$ without creating any $s-t$ paths and, afterwards, move another token to $a$ and obtain $M(s,t)$.
\end{proof}

\begin{lemma}
    \label{lem:root_reconf}
    Let $s,t \in V(G)$ such that $s \in V(G_1)$. Under TJ, any minimal $st$-separator $A$ can be reconfigured into a separator that includes $M(s,t)$.
\end{lemma}

\begin{proof}
    We begin by supposing that $t \in V(G_1)$; since $M(s,t) = V(s,t)$ in this case, an argument similar to the one used in the proof Lemma~\ref{lem:sequential_pair} suffices: for each edge $e_j$ corresponding to vertex $z_j \in V(s,t)$, at least one of its pieces must be in $A$, and all we must do is move it to $z_j$.
    Otherwise, $v \in V(G_1) \setminus \{s\}$.
    If there is some $e_i = st$ and $a \in \supp(t) \setminus \{s\}$, we also proceed as in the proof of Lemma~\ref{lem:sequential_pair}, but now, since $a \in M(s,t)$, the last operation applied is to move a token not on a vertex of $V(s,t)$ to $a$.
    Finally, if there is no $e_i = st$, we proceed as in Lemma~\ref{lem:nested_pair}, where edge $as$ is the support of the first vertex $z \in V(G)$ such that $t$ is a piece of $az$, so that $M(s,t) \{a, z\}$.
\end{proof}

\begin{theorem}
    \label{thm:canon_reconf}
    Under TJ, any two $st$-separators $A$ and $B$ of $G$ can be reconfigured into each other.
\end{theorem}

\begin{proof}
    Let $A' \subseteq A$ and $B' \subseteq B$ be minimal $st$-separators.
    By Lemmas~\ref{lem:parallel_pair} through~\ref{lem:root_reconf}, since both $A'$ and $B'$ can be reconfigured into a separator that contains $M(s,t)$. Let $\mathcal{A'} = \angled{A_1', \dots, A_p'}$ and $\mathcal{B'} = \angled{B_1', \dots, B_q'}$ be such sequences, respectively, such that $A_1' = A'$, $B_1' = B'$, $M(s,t) \subseteq A_p'$ and $M(s,t) \subseteq B_q'$. 
    Let $\mathcal{A} = \angled{A_1, \dots, A_p}$ and $\mathcal{B} = \angled{B_1, \dots, B_q}$ be such that $A_i = A_i' \cup (A \setminus A')$ and $B_i = B_i'\cup (B\setminus B')$.
    It is possible to reconfigure $A_p$ into $B_q$ by moving the tokens in $A_p \setminus M(s,t)$ to $B_p \setminus M(s,t)$ arbitrarily.
    Since, under TJ, the reconfiguration of separators is symmetric, we reconfigure $A$ into $A_p$ through $\mathcal{A}$, $A_p$ into $B_q$, and $B_q$ into $B$ through the reverse of $\mathcal{B}$.
\end{proof}
\section{Concluding Remarks}

In this paper, we investigated the complexity of the reconfiguration of vertex separators under three commonly studied rules: token addition/removal (TAR), token jumping (TJ), and token sliding (TS).
We showed that TAR and TJ are equivalent in the sense that, for every TJ-instance, there is a corresponding TAR-instance, and for every TAR instance, it is either trivially a negative instance, or we can find an equivalent TJ-instance.
We then proceeded to show that \pname{Vertex Separator Reconfiguration} is \NPc\ under TAR/TJ and \PSPACE-\Complete\ under TJ for a subclass of bipartite graphs, which in turn implies that, for general graphs, the three rules are not equivalent to each other unless $\NP = \PSPACE$.
On the positive side, we showed that, for every tame graph class, i.e. with a polynomially bounded number of minimal separators, under TAR/TJ, the reconfiguration problem can be solved in polynomial time; our final results explored classes where this assumption did not hold, namely $\{3P_1, diamond\}$-free -- for which we provide a novel characterization -- and series-parallel graphs, for which we actually prove that it is \textit{always} possible to reconfigure under TAR/TJ.

In terms of future work under TAR/TJ, a natural investigation into the complexity of the problem for different non-tame graph classes is highly desired, particularly because the two examples we have of polynomially solvable cases are trivial in the sense that the answer is always positive, and computing the reconfiguration sequences is achievable by the algorithms we describe in our proofs.
Perhaps there is no gray area and, for a non-tame class, the reconfiguration of vertex separators is either always possible, or answering this query is at least \NPH.
We are specially interested in settling these questions for the two other classes that \citeauthor{Milanic}~\cite{Milanic} showed to be non-tame: $\{claw, K_4, C_4, diamond\}$-free and $\{K_3, C_4\}$-free graphs.
Aside from these jumping rules, while the most obvious one is the study of \pname{Vertex Separator Reconfiguration} under TS, it also appears to be the most challenging, even on restricted classes.
Regardless, it would be quite interesting to find reasonable sufficient conditions or non trivial examples where the reconfiguration problem becomes no harder than \NPc\ under TS.
A related but more particular specific set of queries we would like to answer are about the parameterized complexity of \pname{Vertex Separator Reconfiguration}; for instance, the tight relationship between treewidth and separators may yield useful tools to generalize our result on series parallel graphs to graphs of bounded treewidth.

\begin{comment}
It is also noteworthy that the connection between TS and TJ is quite loose.
To generalize both rules, we can frame the problem as follows.
Given a graph $G$ and a reconfiguration rule $\mathfrak{R}$ that preserves the cardinality of the token set, we can define the \textit{sliding graph} $G_\mathfrak{R} = (V(G), F)$ such that, between two states $A, B$ with $a \in A \setminus B$ and $b \in B \setminus A$, it must hold that $ab \in F$.
Under this point view, if $\mathfrak{R} =$ TS, we have $G_{TS} = G$ and, if $\mathfrak{R} =$ TJ, $G_{TJ}$ is the complete graph.
\end{comment}

\bibliographystyle{plainnat}
\bibliography{main}

\end{document}